\documentclass{IEEEtran}

\IEEEoverridecommandlockouts

\usepackage{cite}
\usepackage{amsmath,amssymb,amsfonts,amsthm,newtxmath}
\usepackage{graphicx,xcolor}
\usepackage[caption=false,font=footnotesize]{subfig}

\newcommand{\expval}[1]{\operatorname{E}[#1]}
\newcommand{\expvalbig}[1]{\operatorname{E}\left[#1\right]}

\newtheorem{lemma}{Lemma}

\begin{document}

\title{A Likelihood Ratio-Based Detector for QTMS Radar and Noise Radar}

\author{
	\IEEEauthorblockN{David Luong,~\IEEEmembership{Graduate Student Member, IEEE}, Bhashyam Balaji,~\IEEEmembership{Senior Member, IEEE}, and \\ Sreeraman Rajan,~\IEEEmembership{Senior Member, IEEE}}

\thanks{D.\ Luong is with Carleton University, Ottawa, ON, Canada K1S 5B6. Email: david.luong3@carleton.ca.}
\thanks{B.\ Balaji is with Defence Research and Development Canada, Ottawa, ON, Canada K2K 2Y7. Email: bhashyam.balaji@drdc-rddc.gc.ca.}
\thanks{S.\ Rajan is with Carleton University, Ottawa, ON, Canada K1S 5B6. Email: sreeraman.rajan@carleton.ca.}
}

\maketitle

\begin{abstract}
	We derive a detector function for quantum two-mode squeezing (QTMS) radars and noise radars that is based on the use of a likelihood ratio (LR) test for distinguishing between the presence and absence of a target. In addition to an explicit expression for the LR detector, we derive a detector function which approximates the LR detector in the limit where the target is small, far away, or otherwise difficult to detect. When the number of integrated samples is large, we derive a theoretical expression for the receiver operating characteristic (ROC) curve of the radar when the LR detector is used. When the number of samples is small, we use simulations to understand the ROC curve behavior of the detector. One interesting finding is there exists a parameter regime in which a previously-studied detector outperforms the LR detector, contrary to the intuition that LR tests are optimal. This is because neither the Neyman-Pearson lemma, nor the Karlin-Rubin theorem which generalizes the lemma, hold in this particular problem. However, the LR detector remains a good choice for target detection.
\end{abstract}


\section{Introduction}

From an abstract, mathematical perspective, radars are machines for hypothesis testing: they decide whether a target is present or absent. The physical operation of the radar dictates the exact nature of the test: what distributions the radar detection data are drawn from, which statistics are being used, and---most importantly---how powerful the test is. In fact, the gold standard for analyzing the detection performance of any radar is the \emph{receiver operating characteristic} (ROC) curve, which gives the probability of detection as a function of the probability of false alarm. That is to say, the ROC curve is the power of the hypothesis test performed by the radar as a function of the significance level. Therefore, when trying to understand how well a given radar works, a careful analysis of the hypothesis testing performed by the radar is fundamental.

In this paper, we are concerned with a class of radars known as \emph{noise radars} \cite{thayaparan2006noise,kulpa2013signal,narayanan2016noise,wasserzier2019noise,cooper1967random}. This type of radar generates a pair of correlated electromagnetic noise signals, of which one is sent at a target and the other is retained as a reference signal within the radar. There is a very important subclass of radars which come under the umbrella of noise radar, namely \emph{quantum two-mode squeezing} (QTMS) radar. As the name implies, QTMS radar is a type of quantum radar, a variant of \emph{quantum illumination} \cite{lloyd2008qi,tan2008quantum,lopaeva2013qi,balaji2018qi}. QTMS radar was the first type of microwave quantum radar for which a laboratory prototype has been built and the results published in scientific publications \cite{chang2018quantum,luong2019quantum,luong2019roc,luong2020magazine}. From the point of view of target detection, it has been shown that noise radars and QTMS radars effectively lie on a continuum characterized by the correlation coefficient between the signal received by a radar and the reference signal stored within that radar \cite{luong2019cov,luong2019rice}. Therefore, the results in this paper will apply to both QTMS radars and noise radars.

In the case of noise radars and QTMS radars, it is typically assumed that the radar detection data is drawn from a multivariate normal distribution with zero mean \cite{dawood2001roc,luong2019cov}. The question then becomes: what test statistic---or \emph{detector function}---should the radar use when conducting the hypothesis test? Various detector functions have been proposed and analyzed in the past \cite{dawood2001roc,luong2019rice,luong2020simdet}, but it appears that one particularly natural test statistic has hitherto escaped notice: the likelihood ratio.

We now present an analysis of the detection performance of QTMS radars and noise radars when the likelihood ratio (LR) is used as a detector function. We explicitly derive an expression for the LR detector function under certain simplifying assumptions. We then derive an approximate LR detector function which is appropriate when the radar attempts to detect a target which is small or far away. Theoretical expressions for the ROC curve are presented for the case for the limit of long radar integration times (large number of samples), and simulation results are shown for short integration times (small number of samples). One major result of our work is that, because the Karlin-Rubin theorem (a generalization of the Neyman-Pearson lemma) does not hold, the LR detector is not necessarily optimal and there exists a parameter regime in which it is known to be outperformed by a previously-studied detector function. Hence, the search for detector functions for QTMS radar and noise radar remains an open problem, though the LR detector is a strong competitor.

\section{Background}
\label{sec:background}

\begin{figure}[t]
	\centerline{\includegraphics[width=.8\columnwidth]{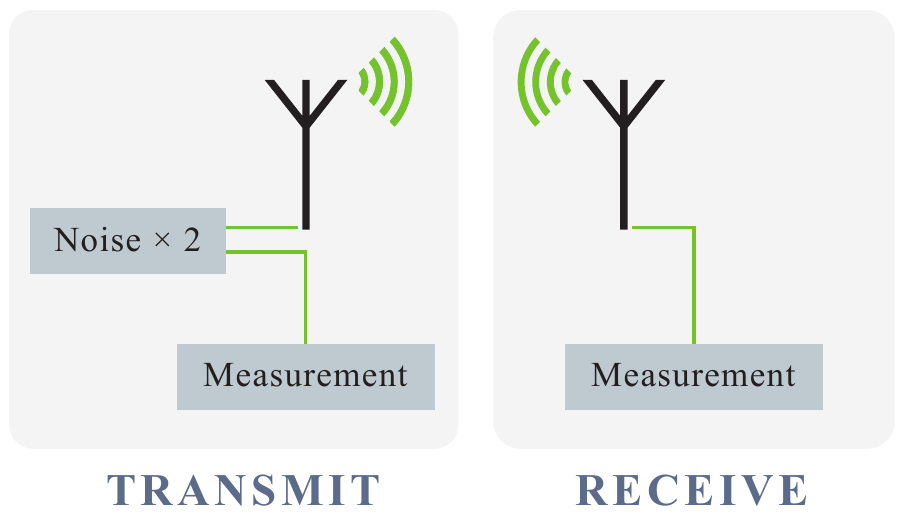}}
	\caption{Block diagram illustrating the basic idea of noise radar. This figure first appeared in \cite{luong2019roc}.}
	\label{fig:block_diagram_abstract}
\end{figure}

Noise radars and QTMS radars work by transmitting an electromagnetic noise signal toward a target, receiving the echo, and comparing the received signal with a reference signal stored within the radar. This is illustrated in Fig.\ \ref{fig:block_diagram_abstract}. Every electromagnetic signal can be mathematically described by two real-valued time series, one representing the in-phase voltages of the signal and the other representing its quadrature voltages. (They can be combined into a single complex-valued time series, but we will not do so in this paper.) In a noise radar, the two voltage time series of the transmitted signal are typically taken to be Gaussian white noise processes with zero mean, and all other noise sources inside and outside the radar are modeled as additive white Gaussian noise (AWGN). It follows that the voltages time series of the received signal are also Gaussian white noise. To our knowledge, all previous work on the theory of noise radar has made use of these assumptions. They are reasonable assumptions in practice, since it is not unreasonable to assume that the properties of the signal generated by the radar are known, nor is it unreasonable to assume that external noise sources are AWGN unless this is specifically known to be false.

Let us denote the in-phase and quadrature voltages of the received signal by $I_1[n]$ and $Q_1[n]$, respectively, where $n$ is the discrete time index. Similarly, let us denote the in-phase and quadrature voltages of the reference signal by $I_2[n]$ and $Q_2[n]$. As a consequence of the discussion above, each of these variables is a Gaussian white noise process with zero mean. They are ``jointly white'' in the sense that $I_1[n_0]$ and $I_2[n_1]$ are independent unless $n_0 = n_1$, and likewise for all other pairs of the four signals. Because this is so, we are only interested in the case where the time difference is zero, and we will suppress the variable $n$ when no confusion arises.

It is evident that the four voltage signals are fully characterized by the $4 \times 4$ covariance matrix $\expval{xx^T}$, where $x = [I_1, Q_1, I_2, Q_2]^T$. In \cite{luong2019cov}, we showed that for a noise radar, $\expval{xx^T}$ can be written in block matrix form as
\begin{equation} \label{eq:QTMS_cov}
	\mathbf{\Sigma} =
	\begin{bmatrix}
		\sigma_1^2 \mathbf{1}_2 & \rho \sigma_1 \sigma_2 \mathbf{R}(\phi) \\
		\rho \sigma_1 \sigma_2 \mathbf{R}(\phi) & \sigma_2^2 \mathbf{1}_2
	\end{bmatrix}
\end{equation}
where $\sigma_1^2$ and $\sigma_2^2$ are the received and reference signal powers, respectively, while $\phi$ is the phase shift between the signals, $\mathbf{1}_2$ is the $2 \times 2$ identity matrix, and $\mathbf{R}(\phi)$ is the rotation matrix 
\begin{equation}
	\mathbf{R}(\phi) = 
	\begin{bmatrix}
		\cos \phi & \sin \phi \\
		-\sin \phi & \cos \phi
	\end{bmatrix} \! .
\end{equation}
QTMS radars are characterized by a similar matrix, but with a reflection matrix instead of a rotation matrix:
\begin{equation}
	\mathbf{R}'(\phi) = 
	\begin{bmatrix}
		\cos \phi & \sin \phi \\
		\sin \phi & -\cos \phi
	\end{bmatrix} \! .
\end{equation}
For the purposes of this paper, it is unimportant whether $\mathbf{R}(\phi)$ or $\mathbf{R}'(\phi)$ is used. The results will be the same in any case.

\subsection{The QTMS/Noise Radar Detection Problem}

Out of the four parameters appearing in the covariance matrix \eqref{eq:QTMS_cov}, it is the correlation coefficient $\rho$ which is of greatest importance in the problem of target detection. When the target is absent, $\rho = 0$ because the received signal is purely background noise which is completely independent of the radar's internal reference signal. Conversely, when the target is present, $\rho > 0$ because some component of the received signal will have originated from the radar, and that component will be correlated with the reference signal. The detection problem for noise radar can therefore be formulated as the problem of distinguishing between the following hypotheses:
\begin{equation} \label{eq:hypotheses}
\begin{alignedat}{3}
	H_0&: \rho = 0 &&\quad\text{Target absent} \\
	H_1&: \rho > 0 &&\quad\text{Target present}
\end{alignedat}
\end{equation}
At this point, we note that the ``quantum advantage'' of QTMS radars over standard noise radars lies in the fact that QTMS radars can achieve higher values of $\rho$. This means that the two hypotheses are easier to distinguish, giving rise to an enhancement in detection performance.

In order to perform this hypothesis test, we must decide on a detector function (test statistic) which allows us to distinguish between the two cases. If we were working with only two time series, a natural choice would be to perform matched filtering. However, we are working with four time series, so we must find a detector which suitably generalizes the matched filter. Previous work has focused on a generalization of matched filtering to complex signals or on estimating $\rho$ directly. In this paper, we use a statistic based on the likelihood ratio in order to derive a detector function.

There are, unfortunately, three other nuisance parameters to be accounted for: $\sigma_1$, $\sigma_2$, and $\phi$. None of these play a direct role in distinguishing whether or not there is a target. However, they considerably complicate the analysis. Therefore, in the derivation of the LR detector, we make the simplifying assumptions $\sigma_1 = \sigma_2 = 1$ and $\phi = 0$. This amounts to a normalization and standardization of the radar's detection data. The time series are normalized to have variance 1, while the reference signal in the radar is ``rotated'' so that there is no phase shift between it and the received signal. We will show via simulations that some of our results hold even when the assumption $\sigma_1 = \sigma_2 = 1$ is violated. As for $\phi = 0$, this assumption was made in previous theoretical and experimental work, so the results given below are in line with previous research. Nevertheless, in future work, we will generalize the LR detector to the case where $\sigma_1$, $\sigma_2$, and $\phi$ are unknown.

With the above discussion in mind, we may formulate the QTMS/noise radar detection problem as follows: let $N$ independent samples be drawn from a multivariate normal distribution with zero mean and covariance matrix
\begin{equation} \label{eq:cov_simplified}
	\mathbf{\Sigma}(\rho) =
	\begin{bmatrix}
		1 & 0 & \rho & 0 \\
		0 & 1 & 0 & \pm\rho \\
		\rho & 0 & 1 & 0 \\
		0 & \pm\rho & 0 & 1 
	\end{bmatrix}
\end{equation}
where the positive sign should be used for noise radars and the negative sign for QTMS radars. (Note that $N$ is related to the integration time $T$ and the sampling frequency $f_s$ of the radar by $T = N/f_s$.) Given these samples, we must decide whether $\rho = 0$ or $\rho > 0$. Our approach will be to use the likelihood ratio to distinguish between these two hypotheses. 

In the calculations that follow, we will assume the use of a QTMS radar and use the negative sign in \eqref{eq:cov_simplified}. To apply our results to standard noise radars, only one sign change is necessary, as we will indicate below.

\subsection{Simulation Procedure}
\label{subsec:simulation}

In this paper, we will have recourse to computer simulations to supplement our theoretical calculations. Naively, we would simply generate $N$ random vectors every time we wish to simulate a single value of the detector function. However, we will see that the likelihood ratio depends on the data only through the sample covariance matrix
\begin{equation} \label{eq:sample_cov}
	\bar{\mathbf{S}} = \frac{1}{N} \sum_{n=1}^N x_n x_n^T
\end{equation}
where $x_n$ is the $n$th sample vector. It is unnecessary, therefore, to generate $N$ random vectors. We need only draw a single random matrix from the Wishart distribution $W_4(\mathbf{\Sigma}, N)$, then normalize the result by $N$. This was the procedure used in \cite{luong2019rice}; it leads to significantly shorter computation times.

In all cases where we obtain ROC curves by simulation, we generate $10^7$ random matrices with $\rho = 0$ and another $10^7$ with a given value of $\rho$ greater than zero. These correspond to the cases where the radar target is absent or present, respectively. From these matrices, we calculate $10^7$ simulated detector function outputs for each case. The ROC curves are then obtained by using the histograms of detector function outputs as empirical probability density functions.

\section{The Likelihood Ratio Detector}

As mentioned previously, one way to test between the hypotheses in \eqref{eq:hypotheses} is to perform a likelihood ratio test. We now derive an explicit expression for the LR detector function.

For $N$ independently drawn samples, the log-likelihood is given by the expression
\begin{equation} \label{eq:log_like}
\begin{split}
	\ell(\rho) = -\frac{N}{2} \left[ \frac{\bar{P}_\text{tot} - \bar{D}_1 \rho}{1 - \rho^2} + 2 \ln (1 - \rho ^2) + 2 \ln(2 \pi) \right]
\end{split}
\end{equation}
where, for brevity, we define
\begin{gather}
	\label{eq:P_tot}
	P_\text{tot} \equiv I_1^2 + Q_1^2 + I_2^2 + Q_2^2 \\
	\label{eq:det1}
	D_1 \equiv I_1 I_2 - Q_1 Q_2.
\end{gather}
A line over an expression indicates the sample mean. For example,
\begin{equation}
	\overline{I_1 I_2} = \frac{1}{N} \sum_{n=1}^N i_1^{(n)} i_2^{(n)}
\end{equation}
where $i_1^{(n)}$ and $i_2^{(n)}$ denote the $n$th samples of $I_1$ and $I_2$, respectively.

As mentioned previously, we have taken the negative sign in \eqref{eq:cov_simplified}. The only change needed when the positive sign is used is to set $D_1 \equiv I_1 I_2 + Q_1 Q_2$.

We can interpret $P_\text{tot}$ as the total power at the radar receiver; it is the sum of the powers of the in-phase and quadrature components of both the received signal and the reference signal. The quantity $D_1$ appeared in \cite{luong2019roc} under the name ``Detector 1'' and was studied in further depth in a previous publication \cite{luong2020simdet}.

The LR detector is defined as
\begin{align} \label{eq:det_LR}
	\hat{D}_\text{LR} &= -2[\ell(0) - \ell(\hat{\rho})] \nonumber \\
		&= N \! \left[ \frac{2\bar{D}_1\hat{\rho} - \bar{P}_\text{tot}\hat{\rho}^2}{1 - \hat{\rho}^2} - 2\ln(1 - \hat{\rho}^2) \right]
\end{align}
where $\hat{\rho}$ is the maximum likelihood estimate of $\rho$. The factor of $-2$ ensures compatibility with Wilks' theorem, which we will invoke in the following section.

In order to obtain $\hat{\rho}$, we must maximize the log-likelihood \eqref{eq:log_like}. To this end, we calculate the derivative of $\ell(\rho)$:
\begin{equation}
	\frac{d\ell(\rho)}{d\rho} = N\frac{\bar{D}_1 - (\bar{P}_\text{tot} - 2)\rho + \bar{D}_1\rho^2 - 2\rho^3}{(1 - \rho^2)^2}.
\end{equation}
We set this equal to 0 in order to find the maximum of $\ell(\rho)$. Assuming that $0 < \hat{\rho} < 1$, the maximum likelihood estimate will be a solution of the cubic equation
\begin{equation} \label{eq:cubic}
	\bar{D}_1 - (\bar{P}_\text{tot} - 2)\rho + \bar{D}_1\rho^2 - 2\rho^3 = 0.
\end{equation}
The general form for the solution of a cubic equation is known, so we may immediately write down the solution. We define the auxiliary quantities
\begin{subequations}
\begin{align}
	A_1 &= 6(\bar{P}_\text{tot} - 2) - \bar{D}_1^2 \\
	A_2 &= (72 - 9\bar{P}_\text{tot} + \bar{D}_1^2) \bar{D}_1 \\
	w &= \sqrt[3]{A_2 + \sqrt{A_1^3 + A_2^2}}, \label{eq:cube_sq_root}
\end{align}
\end{subequations}
in terms of which the maximum likelihood estimate of $\rho$ can be expressed as
\begin{equation} \label{eq:ML_rho}
	\hat{\rho} = \frac{1}{6} \left( w - \frac{A_1}{w} + \bar{D}_1 \right).
\end{equation}
In \eqref{eq:cube_sq_root}, it is necessary to choose the cube and square roots appropriately to ensure $0 < \hat{\rho} < 1$. 

Having obtained this estimate, the LR detector is obtained by simply substituting this value of $\hat{\rho}$ into \eqref{eq:det_LR}.

\subsection{Approximate LR detector}

Because the exact form of the LR detector is unwieldy, we shall derive an approximate form of the detector function. In order to do so, let us assume that $\rho \ll 1$. This regime is of great importance for radar detection because it corresponds to targets which are difficult to detect. For example, they may be far away from the radar, or they could have very small radar cross sections, or the signal-to-noise ratio may be unfavorable. All these cases would lead to a small correlation coefficient \cite{luong2020performance}. Therefore, the assumption $\rho \ll 1$ is not merely a mathematical convenience, but is a very practical one. 

Under the small-$\rho$ assumption, we may expand \eqref{eq:det_LR} in powers of $\rho$, keeping terms up to second order:
\begin{equation} \label{eq:det_LR_approx_maxL}
	\hat{D}_\text{LR} \approx N\hat{\rho} [2 \bar{D}_1 - (\bar{P}_\text{tot} - 2) \hat{\rho}]
\end{equation}
Similarly, we keep terms up to first order in \eqref{eq:cubic} (which is equivalent to keeping terms up to second order in the log-likelihood itself). The resulting equation is linear, and we readily obtain the following approximation of the maximum likelihood estimate $\hat{\rho}$:
\begin{equation}
	\hat{\rho} \approx \frac{\bar{D}_1}{\bar{P}_\text{tot} - 2}.
\end{equation}
By substituting this into \eqref{eq:det_LR_approx_maxL}, we obtain an approximation of the LR detector that is correct up to second order in $\rho$:
\begin{equation} \label{eq:det_LR_approx}
	\hat{D}_\text{LR} \approx \frac{N \bar{D}_1^2}{\bar{P}_\text{tot} - 2}.
\end{equation}
This approximation is clearly easier to calculate than the exact LR detector, an important consideration when processing power is limited.

It is of interest to note that the approximate detector function \eqref{eq:det_LR_approx} can be obtained directly from the exact expressions \eqref{eq:det_LR} and \eqref{eq:ML_rho} by expanding in powers of $\bar{D}_1$ and retaining the lowest-order term. (This can be verified using a computer algebra system such as Mathematica.) Note that $\expval{D_1} = 2\rho$, so when $\rho$ is small, $\bar{D}_1$ may be expected to be small as well.

\section{ROC Curve for the LR Detector}

It follows from Wilks' theorem that, under $H_0$ (target absent), the distribution of the detector in the limit $N \to \infty$ is $D_\text{LR} \sim \chi_1^2$, a chi-square distribution with one degree of freedom \cite{wilks1938large}. The fact that there is one degree of freedom follows from the fact that $\rho$ is a one-dimensional quantity and that the parameter space under $H_0$ is zero-dimensional, being the single point $\rho = 0$.

An approximation to the distribution of $D_\text{LR}$ under $H_1$ (target present) can be obtained by appealing to Theorem 1 of \cite{davidson1970limiting}. That theorem is extremely general and we are only interested in the specific case where the likelihood function contains only one parameter to test and no nuisance parameters. In the following lemma, we restate the theorem for the specific case we are interested in.

\begin{lemma}
	Let $N$ samples be drawn from a family of probability distributions parameterized by $\theta$, let $\delta$ be a constant, and let $\lambda_N$ be the likelihood ratio for testing the following hypotheses:
	\begin{equation}
	\begin{aligned}
		H_0: \theta &= \theta_0 \\
		H_1: \theta &= \theta_0 + \frac{\delta}{\sqrt{N}}.
	\end{aligned}
	\end{equation}
	Finally, let $\mathcal{I}(\theta)$ be the Fisher information for the family of distributions from which the samples were drawn. Then, under the alternative hypothesis $H_1$, it is true that
	\begin{equation}
		-2 \ln \lambda_N \overset{d}{\to} \chi^2_1[\delta^2 \mathcal{I}(\theta_0)].
	\end{equation}
	In other words, as $N \to \infty$, the statistic $-2 \ln \lambda_N$ converges in distribution to the noncentral chi-square distribution with one degree of freedom and noncentrality parameter $\delta^2 \mathcal{I}(\theta_0)$.
\end{lemma}

\begin{proof}
	This is a special case of Theorem 1 of \cite{davidson1970limiting}.
\end{proof}
Note that by taking $\delta = 0$, we immediately recover Wilks' theorem as a special case.

In the case of the LR detector, we take $\theta = \rho$, $\theta_0 = 0$, and $\delta = \rho \sqrt{N}$. It is important to note that, strictly speaking, this proceeding is not legitimate because $\delta$ should not depend on $N$. However, we will show via simulations that for small $\rho$, we still obtain a satisfactory result. To proceed, then, we calculate the Fisher information:
\begin{align}
	\mathcal{I}(\rho) &= -\expvalbig{\frac{\partial^2}{\partial\rho^2} \ell(\rho)} \nonumber \\
		&= \expvalbig{ \frac{-2 + P_\text{tot}(1 + 3\rho^2) - 2 D_1 \rho (3 + \rho^2) + 2\rho^4}{(1 - \rho^2)^3} } \nonumber \\
		&= \frac{1}{(1 - \rho)^2} + \frac{1}{(1 + \rho)^2}.
\end{align}
The final line was obtained by using the linearity of the expectation value operator, inspecting the definitions of $P_\text{tot}$ and $D_1$ in \eqref{eq:P_tot} and \eqref{eq:det1}, respectively, then reading off the appropriate entries of \eqref{eq:QTMS_cov}. It follows from the lemma that, in the limit $N \to \infty$,
\begin{equation}
	D_\text{LR} \sim \chi_1^2 (2 N \rho^2).
\end{equation}

With this result in hand, we can calculate the receiver operating characteristic (ROC) curve for the LR detector. It is known that the cumulative density function for the $\chi_1^2$ distribution is the regularized gamma function $P(1/2, x/2)$. Therefore, the probability of false alarm for a given threshold $T$ is given by the corresponding survival function:
\begin{equation}
	p_\text{FA} = S(T) \equiv 1 - P \left(\frac{1}{2} , \frac{T}{2}\right).
\end{equation}
The survival function for the $\chi_1^2 (2 N \rho^2)$ distribution is the Marcum Q-function $Q_{1/2}(\rho \sqrt{2N}, \sqrt{x})$. It follows immediately that the ROC curve is given by
\begin{equation} \label{eq:ROC_LR}
	p_\text{D}(p_\text{FA}) = Q_{\frac{1}{2}} \! \left[ \rho \sqrt{2N}, \textstyle\sqrt{S^{-1}(p_\text{FA})} \right].
\end{equation}

Because the theorem that we used to derive this ROC curve expression does not strictly hold in our case, it is necessary to verify that this expression does produce reasonable results. Moreover, the approximate equation for the LR detector given in \eqref{eq:det_LR_approx} was based on truncating a series expansion. We verified the validity of both the ROC curve expression \eqref{eq:ROC_LR} and the small-$\rho$ approximation \eqref{eq:det_LR_approx} by running simulations as described in Sec.\ \ref{subsec:simulation}. In Fig.\ \ref{fig:histogram}, we can see that the resulting histograms match up with the theoretical distributions obtained above. In Fig.\ \ref{fig:ROC_LR_1_1}, we plotted the ROC curves that arise from the simulations. The results show that the simulated data agrees extremely closely with the theoretical expression \eqref{eq:ROC_LR}, at least for small values of $\rho$.

\begin{figure}[t]
	\centerline{\includegraphics[width=\columnwidth]{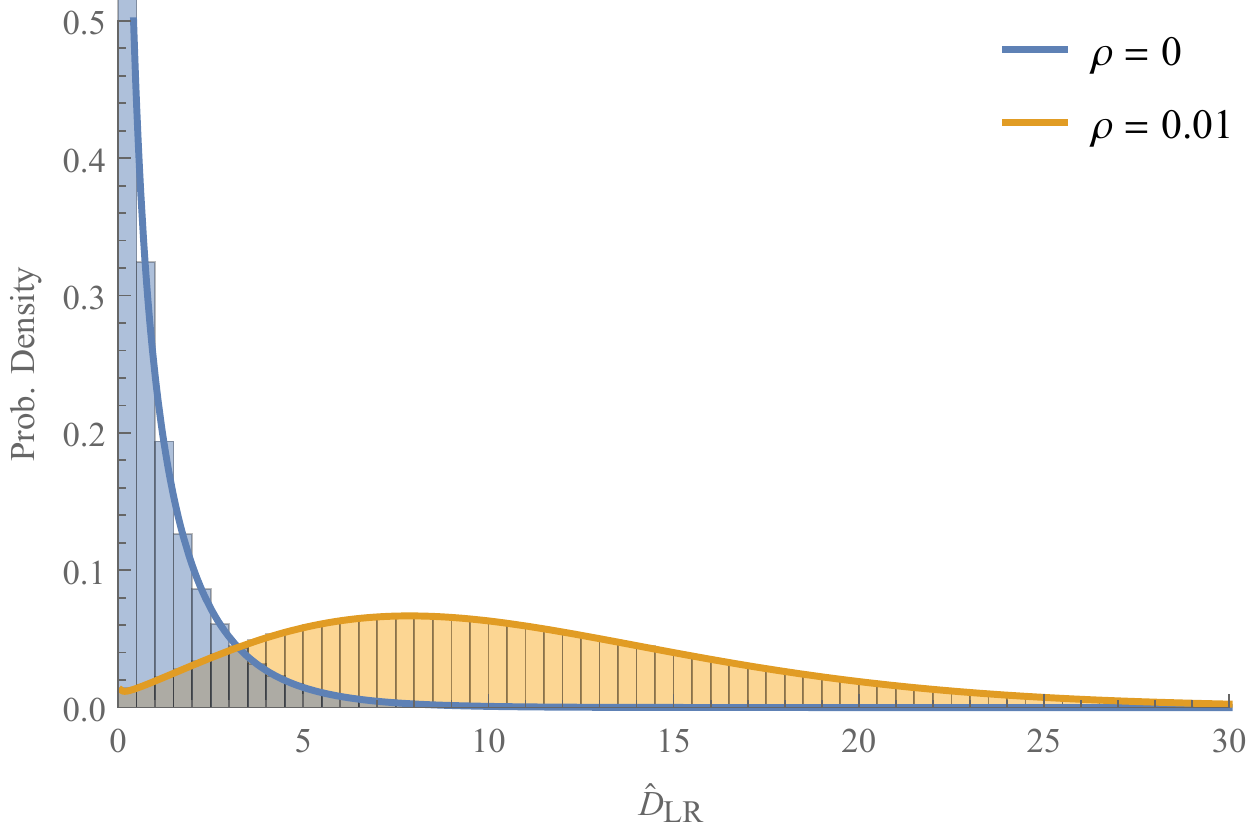}}
	\caption{Histograms of simulated LR detector outputs for $N = 50\,000$, plotted together with theoretical probability density functions. Blue: $\rho = 0$. Orange: $\rho = 0.01$.}
	\label{fig:histogram}
\end{figure}

\begin{figure}[t]
	\centerline{\includegraphics[width=\columnwidth]{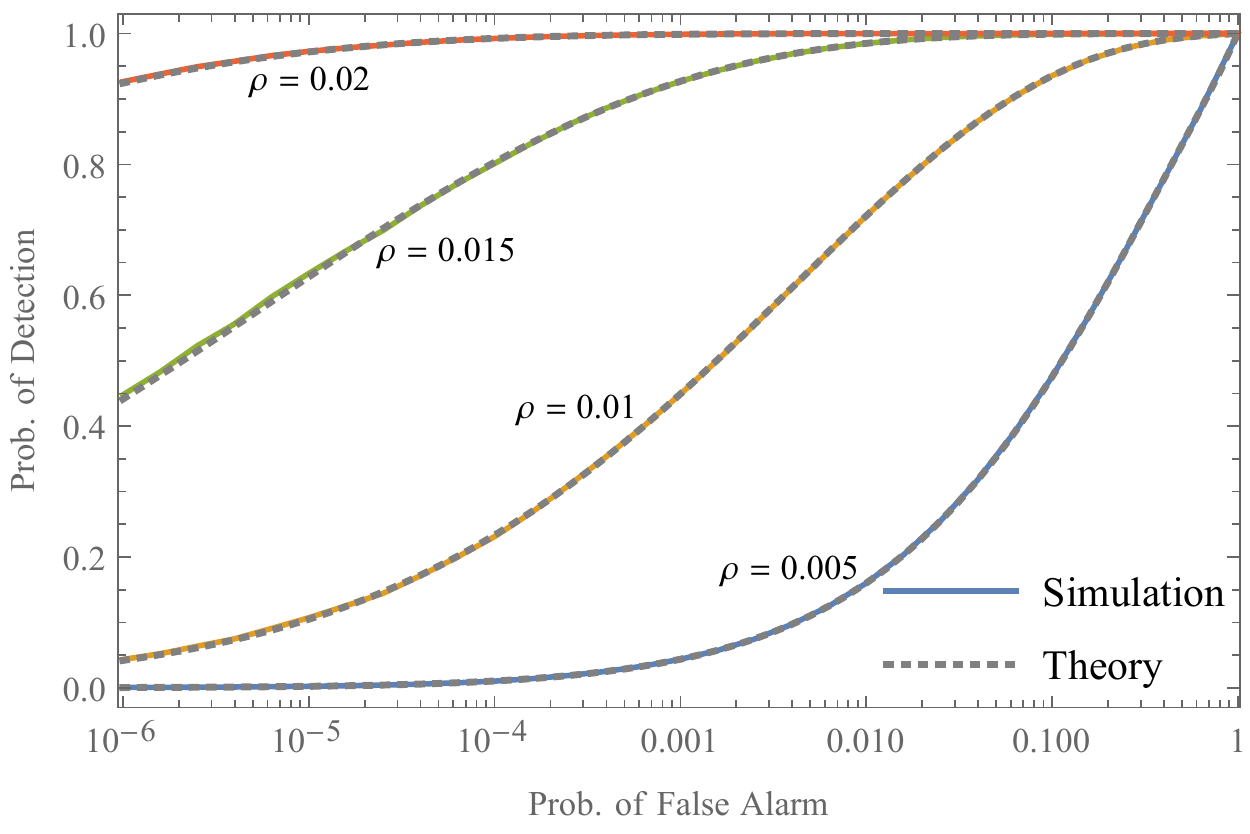}}
	\caption{Comparison of simulated ROC curves for the LR detector with theoretical ROC curves calculated from \eqref{eq:ROC_LR} for $N = 50\,000$ and varying values of $\rho$.}
	\label{fig:ROC_LR_1_1}
\end{figure}

\begin{figure}[t]
	\centering
	\subfloat[]{\includegraphics[width=\columnwidth]{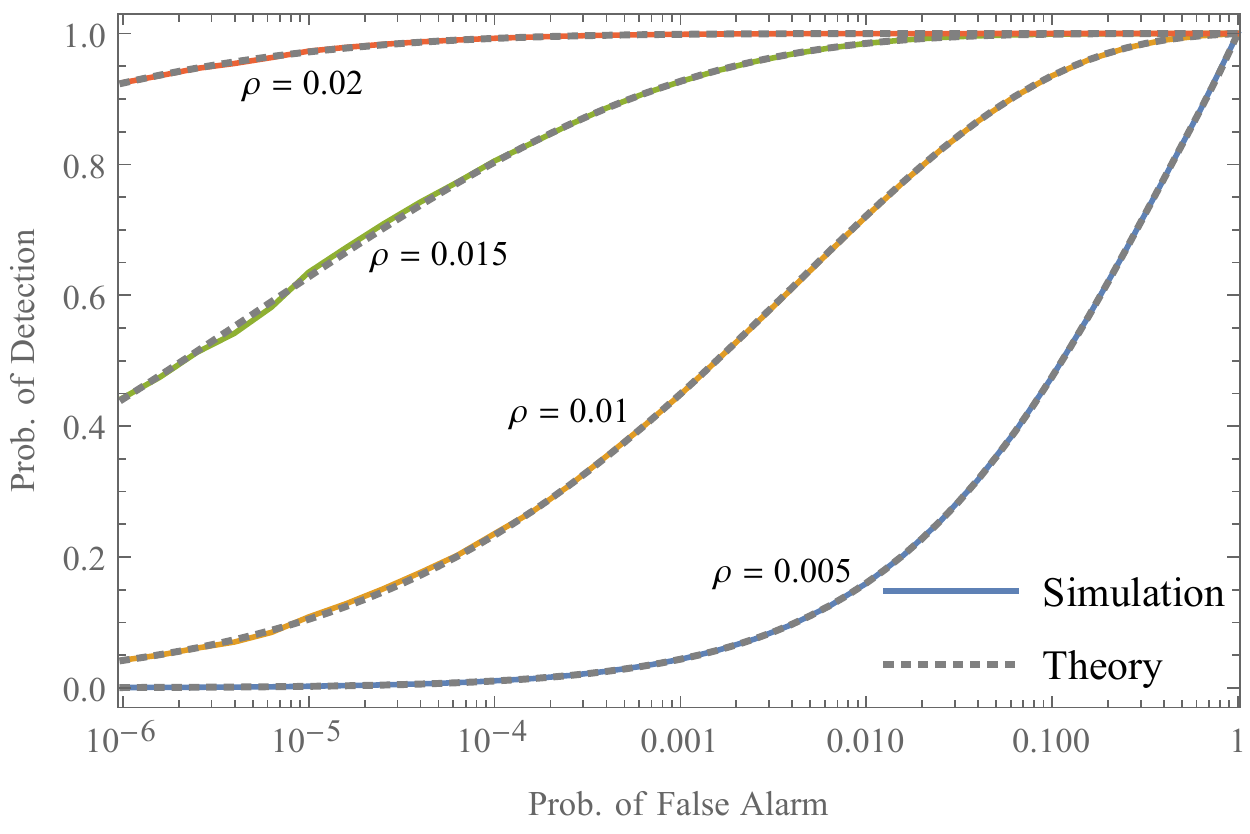}
		\label{subfig:ROC_LR_med_diff}}
	\hfil
	\subfloat[]{\includegraphics[width=\columnwidth]{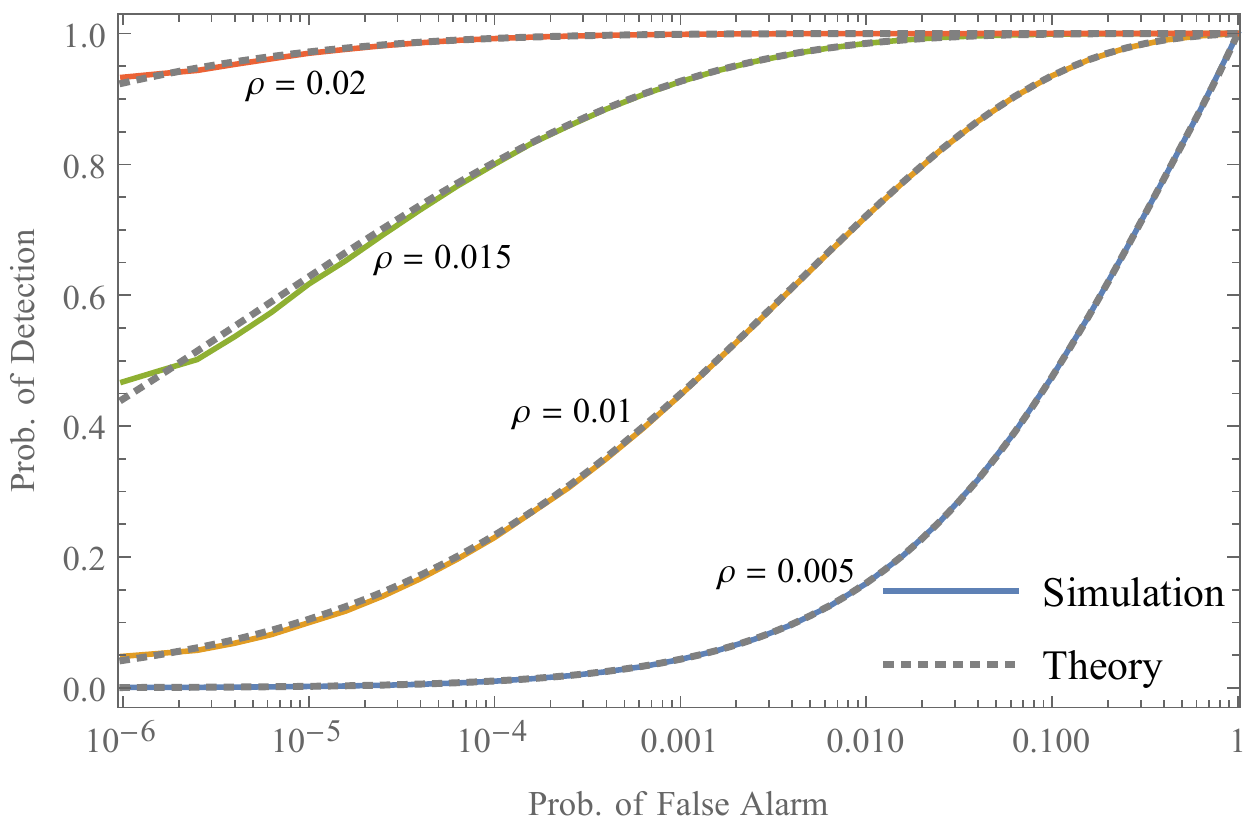}
		\label{subfig:ROC_LR_big_diff}}
	\caption{Comparison of simulated and theoretical ROC curves for the LR detector when $(\sigma_1, \sigma_2)$ equals (a) $(0.1, 10)$ and (b) $(0.01, 10\,000)$. In all cases, $N = 50\,000$.}
	\label{fig:ROC_LR_diff_sigma}
\end{figure}

It is interesting to note that, although the approximate LR detector was derived under the assumption that $\sigma_1 = \sigma_2 = 1$ in \eqref{eq:QTMS_cov}, it is nevertheless still a viable detector function when $\sigma_1$ and $\sigma_2$ are unknown. In Fig.\ \ref{fig:ROC_LR_diff_sigma}, we plot simulated ROC curves for the approximate LR detector when $(\sigma_1, \sigma_2) = (0.1, 10)$ and $(0.01, 10\,000)$. Even in the second case, where there is a large deviation from the assumption $\sigma_1 = \sigma_2 = 1$, the ROC curves do not noticeably deviate from the theoretical expression \eqref{eq:ROC_LR}. From these results, it appears plausible that so long as $\rho \ll 1$ and $\bar{P}_\text{tot} > 2$, the approximate LR detector gives results that are reasonably close to the theoretical expression \eqref{eq:ROC_LR}. (The condition $\bar{P}_\text{tot} > 2$ is necessary because otherwise the denominator of \eqref{eq:det_LR_approx} would become negative.)

\section{Nonoptimality of the LR Detector}

It is well known that, according to the Neyman-Pearson lemma, the LR test is the most powerful one when deciding between two simple hypotheses \cite{casella2002stat}. Referring back to \eqref{eq:hypotheses}, however, we see that the target detection problem does not satisfy the premises of the Neyman-Pearson lemma. The hypothesis that a target is present is not a simple hypothesis. We are not deciding between $\rho = 0$ and $\rho = \rho_0$ for some known value $\rho_0$. It is therefore not permissible to rely on the Neyman-Pearson lemma to state that the likelihood-ratio test is optimal.

There is, however, an extension to the Neyman-Pearson lemma which applies to composite hypotheses of the type seen in \eqref{eq:hypotheses}: the Karlin-Rubin theorem \cite{casella2002stat}. It states that the LR test is the most powerful test for one-sided composite hypotheses such as \eqref{eq:hypotheses} when certain conditions are satisfied. One of these conditions is that there exist a scalar-valued sufficient statistic $T(x)$ for $\rho$. The Fisher-Neyman factorization theorem states that $T(x)$ is a sufficient statistic for $\rho$ if and only if the LR factorizes into the form $h(x) g_\rho[T(x)]$, where $h(x)$ does not depend on $\rho$ while $g_\rho[T(x)]$ depends on the data $x$ only through the statistic $T(x)$. Such a statistic $T$, however, cannot be found. To see this, it is necessary only to inspect the log-likelihood \eqref{eq:log_like}. It is clear that $g_\rho[T(x)]$ satisfies
\begin{equation} \label{eq:g_factor}
	\ln g_\rho[T(x)] = -\frac{N}{2} \left[ \frac{\bar{P}_\text{tot} - \bar{D}_1 \rho}{1 - \rho^2} + 2 \ln (1 - \rho ^2) \right].
\end{equation}
But this expression depends on the data through \emph{two} statistics, namely $\bar{P}$ and $\bar{D}_1$. There is no way to combine $\bar{P}$ and $\bar{D}_1$ so that \eqref{eq:g_factor} is a function of a single scalar-valued statistic for $\rho$. Therefore, the LR does not factor into the form $h(x) g_\rho[T(x)]$ when $T$ is a scalar and there exists no scalar-valued sufficient statistic. It follows that the Karlin-Rubin theorem does not apply, and we cannot rely on this theorem to state that the LR test is optimal.

The above discussion shows that the LR detector is not necessarily the most powerful detector for the target detection problem. It serves as a plausibility argument for the following statement: \emph{the LR detector is not optimal}. The proof is by counterexample: there exists a detector which is in fact more powerful than the LR detector. That detector is none other than $D_1$, which has appeared multiple times in the above calculations.

\begin{figure}[t]
	\centerline{\includegraphics[width=\columnwidth]{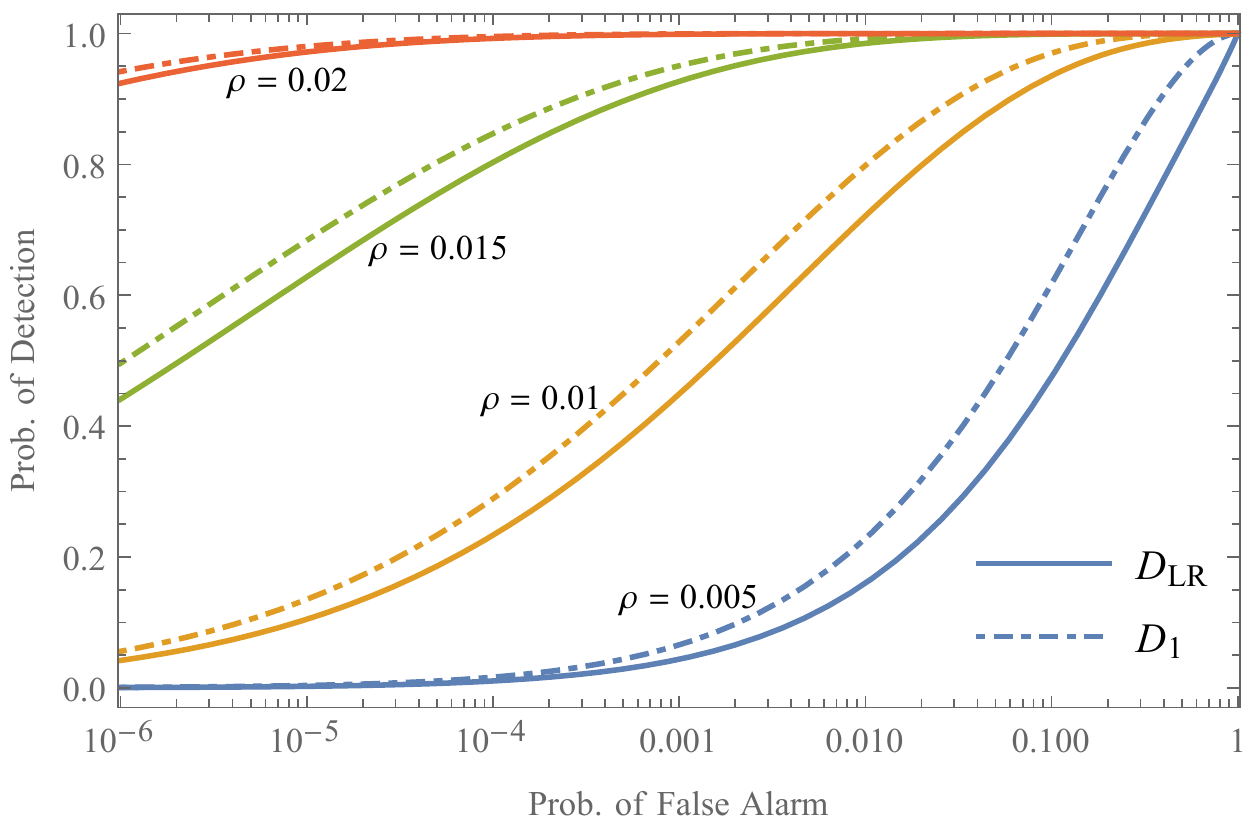}}
	\caption{Comparison of the LR detector with $D_1$ for $N = 50\,000$ and varying values of $\rho$.}
	\label{fig:ROC_LR_D1}
\end{figure}

The use of $D_1$ as a detector function was studied in \cite{luong2019roc}, where the ROC curve in the limit of large $N$ was determined. Thus we may compare that expression with \eqref{eq:ROC_LR}, which is valid in the same limit. Fig.\ \ref{fig:ROC_LR_D1} shows that there is a small but clear advantage in using $D_1$ for target detection compared to using the LR detector. We will show in the following section that $D_1$ is not always better than the LR detector; when $N$ is small, there are regimes where the LR detector is still better than $D_1$. Therefore it is not possible to make blanket statements about the superiority of one detector over another, and the search for good detector functions remains an open problem.

\section{Simulation Results for Small $N$}

The ROC curve expressions in the previous section have all been derived under the assumption that the number of samples $N$ is large. However, in the context of radar detection, it is not always possible to use very large values of $N$. This correspond to the use of long integration times, which is undesirable in situations where radars need to detect targets quickly. Therefore, it is of interest to characterize the detection performance of the LR detector for small $N$. Unfortunately, there are no analytical results that we can rely upon. We therefore turn to simulations.

We note that, when $N$ is small, it is necessary that $\rho$ be made larger in order to compensate. For this reason, the results in this section apply to cases where the target is easier to detect (e.g.\ smaller range or larger radar cross section), and we wish to detect the target very quickly.

\begin{figure}[t]
	\centerline{\includegraphics[width=\columnwidth]{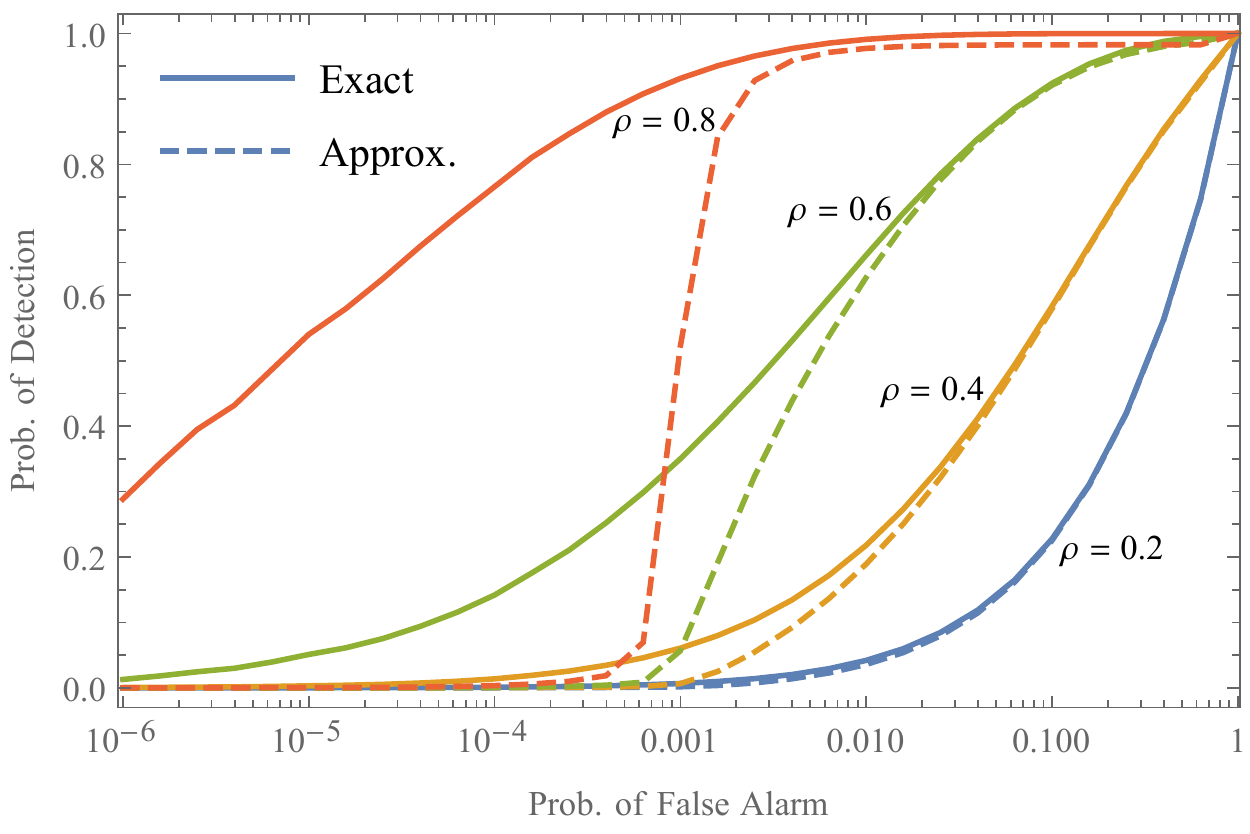}}
	\caption{Comparison of simulated ROC curves for the exact LR detector \eqref{eq:det_LR} with the small-$\rho$ approximation \eqref{eq:det_LR_approx} for $N = 10$ and varying values of $\rho$.}
	\label{fig:ROC_LR_exact_approx_small_N}
\end{figure}

Our first comparison is between the exact LR detector and the approximate detector \eqref{eq:det_LR_approx}. We have shown previously that, even if the assumption that $\sigma_1 = \sigma_2 = 1$ is violated, the approximate detector is still viable. If we could show that it is similarly viable when the assumption $\rho \ll 1$ is violated, we could eschew the use of the exact detector altogether. This would be desirable because the approximate detector requires far less computational requirements. Unfortunately, as shown in Fig. \ref{fig:ROC_LR_exact_approx_small_N}, the approximate LR detector does not perform nearly as well as the exact detector when $\rho$ is high.

\begin{figure}[t]
	\centerline{\includegraphics[width=\columnwidth]{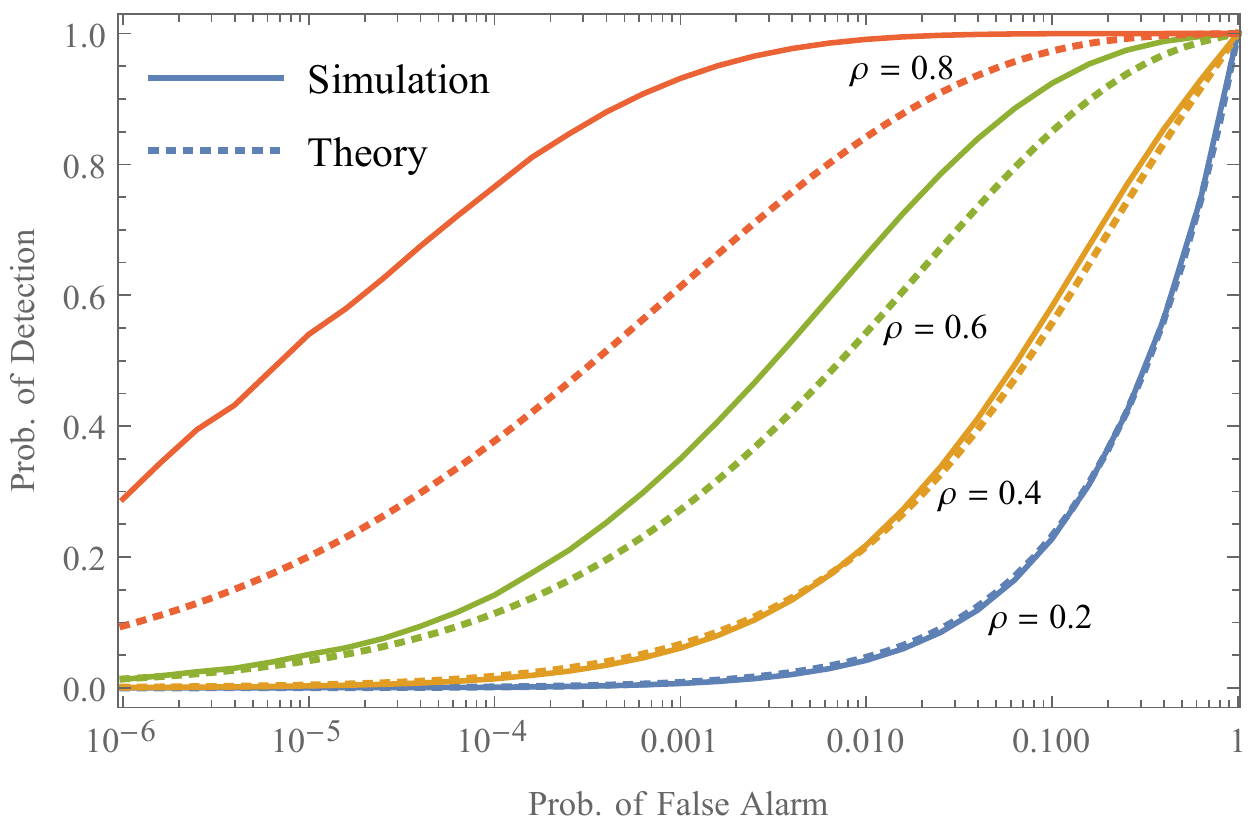}}
	\caption{Comparison of simulated ROC curves for the LR detector with theoretical ROC curves calculated from \eqref{eq:ROC_LR} for $N = 10$ and varying values of $\rho$.}
	\label{fig:ROC_LR_sim_theory_small_N}
\end{figure}

Our next comparison is between simulated ROC curves for the exact LR detector and the corresponding theoretical ROC curves calculated using \eqref{eq:ROC_LR}. The motivation is to see how far we may rely on the theoretical expression even when $\rho$ is large and $N$ is small. Fig.\ \ref{fig:ROC_LR_sim_theory_small_N} shows that, in fact, it does not. Remarkably, when $\rho$ is very large, the simulated ROC curves are actually better than the corresponding theoretical ROC curves. 

\begin{figure}[t]
	\centerline{\includegraphics[width=\columnwidth]{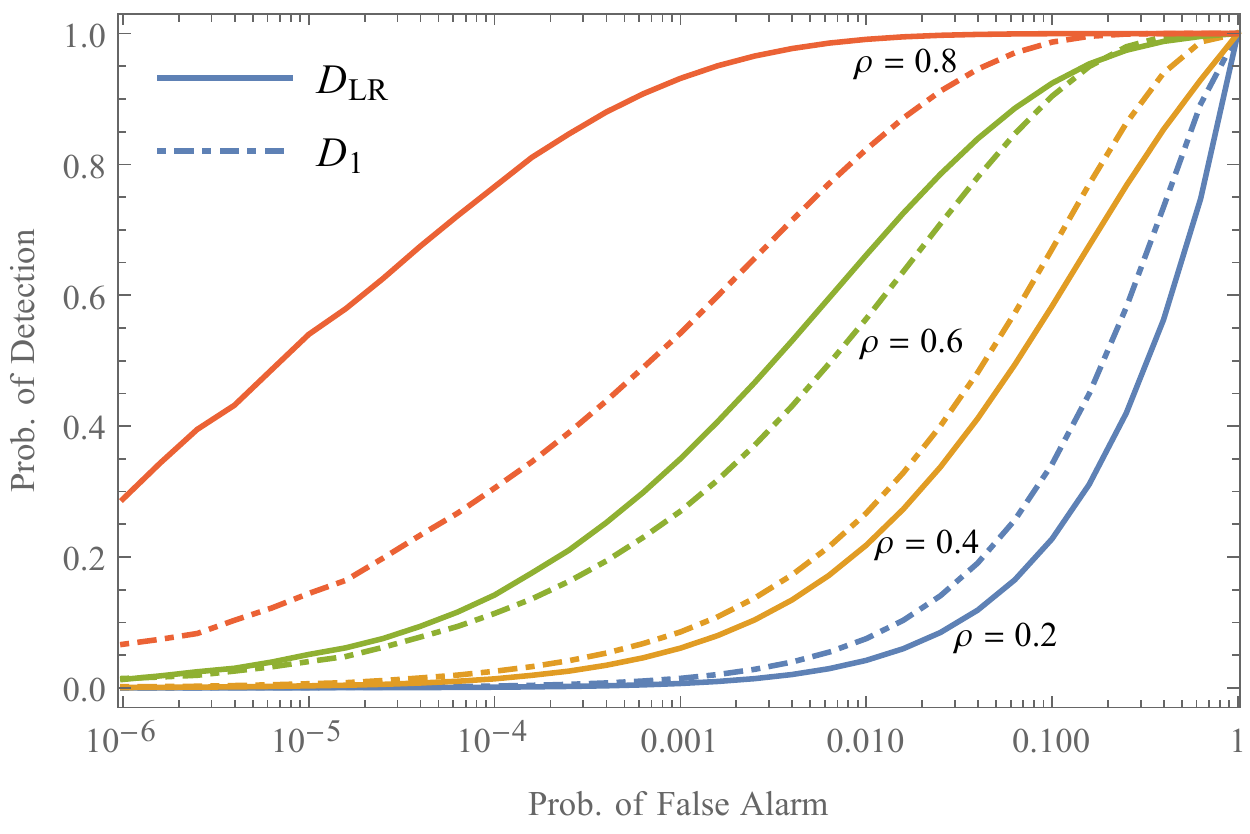}}
	\caption{Simulated ROC curves for the LR detector compared with the $D_1$ detector for $N = 10$ and varying values of $\rho$.}
	\label{fig:ROC_LR_D1_small_N}
\end{figure}

Finally, we compare the performance of the (exact) LR detector with the $D_1$ detector function. As we already saw in Fig.\ \ref{fig:ROC_LR_D1}, the theoretical ROC curve expressions show that $D_1$ is a better detector function when $\rho$ is small. However, the simulation results in Fig. \ref{fig:ROC_LR_D1_small_N} suggest that there is a crossover in detection performance, and when $\rho$ is high, $D_\text{LR}$ is better than $D_1$. Therefore, we are not justified in making blanket statements about the optimality of any detector function. Instead, it may be a better strategy to choose $D_\text{LR}$ when $\rho$ is expected to be large, and $D_1$ when $\rho$ is expected to be small.

\section{Conclusion}

In this paper, we derived and analyzed a detector function for QTMS radar and noise radar that is based on the likelihood ratio. In addition to an exact formula for the LR detector, we derived an approximate detector function which holds in the limit of small correlation coefficients $\rho$. We derived a mathematical formula for the ROC curve of the LR detector which holds when the number of samples $N$ is large; we also performed simulations to understand the behavior of the detector when $N$ is small. We found that the LR detector is not optimal, which runs counter to the intuition that the likelihood ratio test is always optimal; it is outperformed by the $D_1$ detector function when $\rho$ is small. However, the LR detector is a strong competitor, significantly outperforming $D_1$ when $\rho$ is very large and $N$ is small.

One drawback of our work is that it assumes that the values of the nuisance parameters $\sigma_1$, $\sigma_2$, and $\phi$ are known. Although there exist situations where this is a reasonable assumption, it would be more satisfying if we could derive the LR detector function without the use of these assumptions. This will be the subject of future work.

\section*{Acknowledgment}

D.\ Luong acknowledges Ian Lam Wheng-Kit for valuable discussions and for bringing \cite{davidson1970limiting} to our attention.

\bibliographystyle{ieeetran}
\bibliography{qradar_refs,own_refs}

\end{document}